%% file: UG-main.tex
\documentclass[11pt,a4paper]{article}

\usepackage[a4paper, total={5.8in, 9in}]{geometry}

\usepackage[english]{babel}
\usepackage{mdframed}

\usepackage{hyperref}
\usepackage{bbm}
\usepackage{enumerate}
\usepackage{amssymb}
\usepackage{amsmath}
\usepackage{alltt}
\usepackage{amsthm}
\usepackage{amsfonts}
\usepackage{mathtools}
\usepackage{latexsym}
\usepackage{xspace}
\usepackage{complexity}
\usepackage{todonotes}
\usepackage[ruled,vlined]{algorithm2e}			

\usepackage{thm-restate}
\usepackage[normalem]{ulem}

\usepackage{epsfig}
\usepackage{graphicx}

\newcommand{\bd}{\bar{d}}

\usepackage{xcolor} 

\newcommand{\minLinEqQ}{\textsc{Min-Lin-Eq$(q)$}\xspace}
\newcommand{\maxLinEqQ}{\textsc{Max-Lin-Eq$(q)$}\xspace}
\newcommand{\minLinEqQcomp}{\textsc{Min-Lin-Eq$(q)$-Full}\xspace}

\newcommand{\UG}{\textsc{Min-Unique-Games$(q)$}\xspace}
\newcommand{\UGMax}{\textsc{Max-Unique-Games$(q)$}\xspace}
\newcommand{\UGC}{\textsc{Min-Unique-Games$(q)$-Full}\xspace}
\newcommand{\UGCMax}{\textsc{Max-Unique-Games$(q)$-Full}\xspace}

\newcommand{\FAST}{\textsc{Feedback-Arc-Set-Tournaments}\xspace}
\newcommand{\corc}{\textsc{Correlation-Clustering}\xspace}

\newcommand{\TMP}{\ensuremath{\operatorname{TEMP}}\xspace}
\newcommand{\FIN}{\ensuremath{\operatorname{FINAL}}\xspace}
\newcommand{\VAL}{\ensuremath{\operatorname{VAL}}\xspace}

\newcommand*{\opt}{\ensuremath{\operatorname{OPT}}\xspace}
\newcommand*{\OPT}{\ensuremath{\operatorname{OPT}}\xspace}
\newcommand*{\optval}{\ensuremath{\operatorname{OPT_{val}}}\xspace}

\newenvironment{cproof}
{\begin{proof}
 [Proof.]
 \vspace{-1.5\parsep}
}
{ \end{proof}}

\newtheorem{theorem}{Theorem}
\newtheorem{lemma}[theorem]{Lemma}

\newtheorem{observation}[theorem]{Observation}
\newtheorem{claim}[theorem]{Claim}

\newtheorem{problem}{Problem}
\numberwithin{theorem}{section}

\newcommand{\Ex}{\mathbbm{E}}
\newcommand{\I}{{\mathcal{I}}}
\newcommand{\T}{{\mathcal{T}}}

\newcommand{\eps}{\varepsilon}

\newcommand{\deltUp}{{\frac{c \cdot q \log{q}}{\ell}}}
\newcommand{\Gstar}{{G^k_{\star}}}

\newcommand{\Gm}{{G^2_{mult}}}
\newcommand{\Gs}{{G^2_{\star}}}

\begin{document} \bibliographystyle{alpha}

\title{Voting algorithms for unique games on complete graphs}

\author{Antoine M\'eot\thanks{Laboratoire G-SCOP (CNRS, Grenoble-INP),
    Grenoble, France} \and Arnaud
  de Mesmay\thanks{LIGM, CNRS, Univ. Gustave Eiffel, ESIEE Paris,
    F-77454 Marne-la-Vall\'ee, France}
\and Moritz
  M\"uhlenthaler\footnotemark[1]
\and Alantha Newman\footnotemark[1]}

\date{}

\maketitle

\input{abstract}

\input{introduction}

\input{pivot-algorithm}

\input{voting-algorithm}

\input{PTAS}

\input{NPhardness}

\bibliography{unique-games}

\input{appendixA}

\input{appendixB}

\input{appendixhardness}

\end{document}

%% file: abstract.tex
\begin{abstract}
An approximation algorithm for a constraint satisfaction problem is
called {\em robust} if it outputs an assignment satisfying a $(1 -
f(\epsilon))$-fraction of the constraints on any
$(1-\epsilon)$-satisfiable instance, where the loss function $f$ is
such that $f(\epsilon) \rightarrow 0$ as $\epsilon \rightarrow 0$.
Moreover, the runtime of a robust algorithm should not depend in any
way on $\epsilon$.  In this paper, we present such an algorithm for
\UG on complete graphs with $q$ labels.  Specifically, the loss
function is $f(\epsilon) = (\epsilon + c_{\epsilon} \epsilon^2)$,
where $c_{\epsilon}$ is a constant depending on $\epsilon$ such that
$\lim_{\epsilon \rightarrow 0} c_{\epsilon} = 16$.  The runtime of our
algorithm is $O(qn^3)$ (with no dependence on $\epsilon$) and can run
in time $O(qn^2)$ using a randomized implementation with a slightly
larger constant $c_{\epsilon}$.  Our algorithm is combinatorial and
uses voting to find an assignment.  It can furthermore be used to
provide a PTAS for \UG on complete graphs, recovering a result of
Karpinski and Schudy with a simpler algorithm and proof. We also prove
\NP-hardness for \UG on complete graphs and (using a randomized
reduction) even in the case where the constraints form a cyclic
permutation, which is also known as {\sc Min-Linear-Equations-mod-$q$}
on complete graphs.
\end{abstract}

%% file: introduction.tex
\section{Introduction}

As defined by Zwick~\cite{zwick1998finding}, an approximation
algorithm for a constraint satisfaction problem (CSP) is called {\em
  robust} if it outputs an assignment satisfying a $(1 -
f(\eps))$-fraction of the constraints on any $(1-\eps)$-satisfiable
instance, where the loss function $f$ is such that $f(\eps)
\rightarrow 0$ as $\eps \rightarrow 0$.  Moreover, the runtime of the
algorithm should not depend in any way on $\eps$.  Robust algorithms
for CSPs have been studied
extensively~\cite{guruswami2011tight,kun2012linear,barto2016robustly,dalmau2019robust}.
For example, the famous random hyperplane rounding algorithm for the maximum cut
problem yields a robust approximation for the complementary
minimization problem~\cite{goemans1995improved} and is essentially
optimal~\cite{o2008optimal}.

Let us call an approximation algorithm {\em super robust} if the loss
function has the form $f(\eps) = \eps + O(\eps^2)$. Such super robust
algorithms are relevant in the design of approximation algorithms
because, as we will discuss later on, if one has a super robust
algorithm for the min version of a problem and a polynomial time
approximation scheme (PTAS) for the complementary max version, then we
can derive a PTAS for the min version as well.  Note that the
existence of a PTAS does not imply the existence of a super robust
algorithm.  There is a wide range of techniques to obtain a PTAS for
the max versions of various constraint satisfactions problems on dense
graphs (see e.g.,~\cite{arora1999polynomial}).  In contrast, we are
not aware of super robust algorithms for CSPs or similar problems,
even on dense graphs.

In this article we investigate super robust approximation algorithms
for {\sc Unique-Games} on complete graphs, which are CSPs.  We now
define the problems under consideration. Let $G=(V,E)$ be a complete
graph with an arbitrary linear order on the vertices, let $q$ be a
positive integer (where $q \leq \poly(n)$) and let $[q] = \{0, \ldots,
q-1\}$.  (Note that $G$ is simple and therefore does not contain any
multi-edges or self-loops.)  Let $n$ denote the number of vertices and
$m$ the number of edges in $G$ (i.e., $n=|V|$ and $m = {n \choose
  2}$).  We use $uv=vu$ to refer to an edge in $E$ and $(u,v)$ to
refer to an ordered pair or arc.  An assignment is a map $x: V
\rightarrow [q]$ giving a label $x_v$ to each vertex $v$. For each
ordered pair of vertices $(u, v)$ there is a permutation $\pi_{uv}:
[q] \rightarrow [q]$.  This permutation is interpreted as a constraint
as follows: an assignment $x$ satisfies the constraint if
$x_v=\pi_{uv}(x_u)$. This is equivalent to the constraint
$x_u=\pi_{vu}(x_v)$ since we require $\pi_{vu} = \pi^{-1}_{uv}$. A set
of constraints is satisfiable if there exists an assignment satisfying
all of them. Then the \UGC problem is the following.

\begin{problem}[\UGC]\label{UG}
Given a complete graph $G$, a positive integer $q$ and a permutation
$\pi_{uv}:[q] \rightarrow [q]$ for each ordered pair of vertices
$(u,v)$ with $u < v$ (such that $\pi_{vu} = \pi^{-1}_{uv}$), find a
minimum cardinality subset of edges of $G$ whose deletion results in a
satisfiable set of constraints.
\end{problem}
In a special case of this problem, each permutation is {\em cyclic}.
Specifically, for each ordered pair of vertices $(u,v)$ there is a
given integer $c_{uv} \in [q]$ (symmetrically, $c_{vu} = q - c_{uv}
\bmod q$).  For each edge $uv \in E$ with $u < v$, there is a
constraint $x_u - x_v \equiv c_{uv} \bmod q$.  (Observe that $x_v -
x_u \equiv c_{vu} \bmod q$ is an equivalent constraint.)  In general
graphs, this problem is also known as {\sc
  Min-Linear-Equations-mod-$q$}, which we abbreviate to {\sc
  Min-Lin-Eq$(q)$}.

\begin{problem}[\minLinEqQcomp]\label{linEqModq}
Given a complete graph $G$, a positive integer $q$ and a constraint
$x_u - x_v \equiv c_{uv} \bmod q$ for each ordered pair of vertices
$(u,v)$ with $u < v$ (such that $c_{vu} = q - c_{uv}$), find a minimum
cardinality subset of edges of $G$ whose deletion results in a
satisfiable set of constraints.
\end{problem}

We refer to the general versions of Problems 1 and 2 (i.e., when $G$
is not necessarily a complete graph) as \UG and \minLinEqQ,
respectively, and to the complementary versions (i.e., when one aims at maximizing the number of satisfied constraints) as \UGMax and \maxLinEqQ.
Although it might seem like an easier problem, a constant factor approximation for \maxLinEqQ yields a constant factor approximation for \UGMax~\cite{khot2007optimal}.

\subparagraph*{Our results.} In this paper, we first present a super robust
algorithm for \minLinEqQcomp.  Specifically, the runtime of our
algorithm is $O(q n^3)$ in the RAM model (with no dependence on
$\eps$) and the loss function is $f(\eps) = (\eps + c_{\eps} \eps^2)$,
where $c_{\eps}$ is a constant depending on $\eps$ such that
$\lim_{\eps \rightarrow 0} c_{\eps} = 16$.  A randomized
implementation with a slightly larger constant $c_{\eps}$ in the loss
function runs in time $O(q n^2)$.  We show that our algorithm can be
extended to the so-called {\em everywhere dense} case, which is where
every vertex has degree at least $(1-\delta)(n-1)$ for some constant
density parameter $\delta \in (0,1)$~\cite{arora1999polynomial}.

Our algorithm is very simple, purely combinatorial and uses voting to
find an assignment.  First, we find an initial assignment using a
pivot algorithm in the spirit of \cite{ailon2008aggregating}, which is a 
3-approximation in the case of \minLinEqQcomp
(Section \ref{sec:pivot}).  Then we improve this solution
according to ``votes'' of the
other vertices based on their initial assignments (Section
\ref{sec:voting}).  We discuss the extension to the dense case, whose
details can be found in Appendix \ref{app:B}. When the alphabet size
is constant, we can couple our robust algorithm with classical
approximation algorithms for the complementary problem to obtain a
PTAS for \UGC (and thus \minLinEqQcomp). This is explained in
Section~\ref{sec:ptas}, and recovers a result of Karpinski and
Schudy~\cite{karpinski2009linear}, with a simpler proof.  Recall
that given a $(1-\eps)$-satisfiable instance, we can find a $(1+
O(\eps))$-approximation via an algorithm whose running time is
independent of $\eps$.  To obtain such a guarantee via the algorithm
of Karpinski and Schudy, we would need to exhaustively search for an
assignment on a sample of size $\Omega(1/\eps^2)$, which leads to a
running time of $\Omega(q^{1/\eps^2})$.  Thus finding an algorithm
that skips this exhaustive assignment step typical of a PTAS is the
key to obtaining a super robust algorithm.

  We also consider the hardness of \UGC (Section~\ref{sec:nphard}).
  In the case of $q=2$, the \NP-hardness for \minLinEqQcomp follows
  from the NP-hardness of \corc with two clusters (i.e., {\sc
    MinDisAgree[2]}) due to Giotis and
  Guruswami~\cite{giotis2006correlation}.  For $q \geq 3$, the
  hardness of \UGC does not appear to be explicitly considered
  anywhere in the literature and thus its complexity status was open.
  Therefore, we prove \NP-hardness for \UGC for $q \geq 3$.  For
  \minLinEqQcomp, we prove NP-hardness under the weaker assumption
  that $\NP \subsetneq \BPP$.  Our reduction is similar to the
  hardness reductions for
  \FAST~\cite{ailon2008aggregating,alon2006ranking,charbit2007minimum}
  and for fully-dense problems~\cite{ailon2007hardness} but is not
  directly implied by them since, for example, the latter result only
  holds for fully-dense CSPs on a binary domain. Both proofs are
  deferred to Appendix~\ref{app:hardness}.

\subparagraph*{Background on Unique-Games.}
{\sc Unique-Games} is one of the most important
problems in approximation algorithms due to its direct connection with
the famous Unique Games Conjecture of Khot~\cite{khot2002power}, which
has wide-ranging implications in the hardness of approximation.
Roughly speaking, the conjecture states that there is no
constant-factor approximation algorithm for \UGMax.  It is not hard to
see that there is an algorithm with approximation factor $1/q$.  Many
approximation algorithms, which beat this factor, have been developed,
although none give constant-factor approximations.  Some of these use
semidefinite programming
(SDP)~\cite{khot2002power,trevisan2005approximation,charikar2006near,raghavendra2008optimal},
and some use linear programming (LP)~\cite{gupta2006approximating}.
It is known that one can find a constant factor approximation for
\UGMax in subexponential
time~\cite{arora2015subexponential,barak2011rounding,bafna2021playing}.
See \cite{khot2010inapproximability,steurer2014sum} for surveys on the
Unique Games Conjecture.

\subparagraph*{Applications.} In addition to its theoretical significance, {\sc Unique-Games} is
closely related to angular synchronization and phase reconstruction
problems with applications in many fields including computer
vision~\cite{agrawal2006range} and optics~\cite{walther1963question,
millane1990phase,rubinstein2001reconstruction}.
The models considered in these applied settings are usually
constructed by fixing a satisfiable instance and adding noise from
some specified distribution to each constraint~\cite{boumal2013MLE,bandeira2017tightness,zhong2018near,gao2019multi,iwen2020phase}.  (This corresponds to
perturbing each $c_{uv}$.)  The goal is exact recovery of the original
(satisfiable) instance.
Another, more combinatorial, model corresponds more closely to the
statement of the {\sc Unique-Games} problem.  In this setting, we
begin with a satisfiable instance and for each constraint, with some
specified probability, noise from a known distribution is added
~\cite{singer2011angular}.
Notice that in this setting, not all constraints are necessarily
perturbed.  Thus, under certain parameters (e.g., small probability of
perturbing a constraint and uniform noise), the solution to the
original input instance is the solution to the instance of {\sc
  Unique-Games} problem corresponding to the perturbed instance.
Both models have been studied on complete graphs~\cite{singer2011angular,gao2019multi,filbir2021recovery}.  
Since the noise is
generated from some particular distribution, the problem instance is
not a worst-case or adversarially perturbated instance, and the
analysis of the recovery procedures usually requires knowledge of the
specific perturbation model.  Nevertheless, an algorithm with a
worst-case performance guarantee, such as ours, can be applied to
instances belonging to this model. 
In practical settings,
the simplicity and implementability of our algorithm are desirable
properties.

\subparagraph*{Previous results.}
In terms of a robust algorithm for \UG, there is an algorithm based on
semidefinite programming with loss function $f(\eps) = \sqrt{\eps
  \log{q}}$~\cite{charikar2006near}.  This is not really a robust
algorithm for \UG since the loss function depends on $q$ and not
solely on $\eps$, and $q$ could be a function of $n$.  Robust
algorithms for Constraint Satisfaction Problems have been studied in
depth~\cite{guruswami2011tight,kun2012linear,barto2016robustly,dalmau2019robust}.
\UG has also been studied on expanders \cite{arora2008unique}, and
this work gives an algorithm with loss function
$O(\frac{\eps}{\lambda} \log{\frac{\lambda}{\eps}})$, where $\lambda$
is the second smallest eigenvalue of the normalized Laplacian of the
input graph $G$.  This algorithm is robust in the case of complete
graphs, since $\lambda = 1$ for a complete graph.  In
\cite{guruswami2011lasserre}, the stated loss function for a graph
with $\lambda = 1$ is $f(\eps) = (3 + \eta)\eps$, which is achieved in
time $n^{O(2/\eta)}$.  Perhaps a more careful analysis of these
algorithms can yield a slightly better loss function in the case of
complete graphs.  In any case, these loss functions correspond to
constant factor approximations and
are therefore worse than the one we
present in this paper by an order of magnitude, and for example, they
cannot be leveraged to obtain a PTAS as in Section \ref{sec:ptas}.
Moreover, it is somewhat interesting that our loss function can be
achieved using combinatorial methods rather than relying on tools from
semidefinite programming as is the case in~\cite{arora2008unique} and
on semidefinite hierarchies as in~\cite{guruswami2011lasserre}.  We
remark that the Unique Games Conjecture is equivalent to the
conjecture that a basic assignment-based semidefinite program is the
best tool for solving an instance of
\UGMax~\cite{raghavendra2008optimal}.  Thus, it is reasonable to
consider different algorithmic tools.  We note that the algorithm of
\cite{arora2008unique} can be interpreted as a pivot algorithm and we
discuss this connection in Section \ref{sec:pivot}.
    
\UG has also been studied on dense graphs and there is a PTAS with stated
runtime $O(n^2) + 2^{O(\frac{1}{\eps})}$~\cite{karpinski2009linear}.
This algorithm, based on a combination of random sampling and voting,
is not robust as the runtime depends on $\eps$.  Notice
that this runtime assumes that both $q$ and the density parameter
$\delta$ are fixed (i.e., the dependence on $q$ and $\delta$ occur in
the exponent but are not stated explicitly in the runtime).

Finally, we remark that many combinatorial optimization problems have
been specifically studied on complete graphs or tournaments.  For
example, \FAST has a much better approximation guarantee than is
currently known for the general case, but is still
\NP-hard~\cite{ailon2008aggregating}.  Another well-studied example is
the special case of \corc known as {\sc MinDisagree} on complete
graphs~\cite{bansal2004correlation,charikar2005clustering,giotis2006correlation,ailon2008aggregating,chawla2015near}.
The latter problem is APX-hard~\cite{charikar2005clustering}, so it is
unlikely to have a super robust approximation (see Section~\ref{sec:ptas}).  Although \FAST has a
PTAS~\cite{kenyon2007rank}, 
it also does not seem to have a known super
robust approximation algorithm.

%% file: pivot-algorithm.tex
\section{Pivot Algorithm for \minLinEqQcomp}\label{sec:pivot}

In a given instance of \minLinEqQcomp on a graph $G$, each cycle
in $G$ is either {\em consistent} or {\em inconsistent}.  A cycle is
consistent (inconsistent) if it is satisfiable (unsatisfiable,
respectively).  Observe that a feasible solution to Problem
\ref{linEqModq} is a hitting set for the set of inconsistent cycles.
The following algorithm outputs a vertex labeling such that the
unsatisfied edges form a hitting set for the inconsistent cycles.

\vspace{5mm}
\noindent
\fbox{\parbox{15cm}{

{\bf Pivot Algorithm} 

\vspace{1mm}

{\it Input:} An instance of \minLinEqQcomp on a graph $G=(V,E)$.

\begin{itemize}

\item[1.] Pick a pivot $p \in V$ uniformly at random and label $p$
  with 0.

\item[2.] For each vertex $v \in V\setminus{p}$, assign $v$ label
  corresponding to the constraint on edge $pv$.  (Specifically,
  $\ell(v) = c_{vp}$.)

\end{itemize}
}}

\vspace{3mm}

On an input for \UGC, the algorithm can be modified to test each
possible label in $[q]$ for the pivot chosen in Step 1.

\begin{restatable}{theorem}{SimpleDualProof}\label{thm:3approx}
The {Pivot Algorithm} is a $3$-approximation algorithm for
Problem \ref{linEqModq}.
\end{restatable}

The proof of Theorem \ref{thm:3approx} follows almost directly from
the analysis of the {\bf KwikSort Algorithm} for
\FAST~\cite{ailon2008aggregating}.  For completeness, the proof can be
found in Appendix \ref{app:first}.  We also give an example showing
that this analysis is tight.  We remark that for a satisfiable instance of {\sc Unique-Games}, one can choose any spanning tree and ``propagate'' the values along the spanning tree, resulting in an optimal solution.  The pivot algorithm is also a type of spanning-tree algorithm, since it determines the assignments by using the edges incident to the pivot, which form a star-shaped spanning-tree.

\subsection{Pivot Algorithm and SDP Rounding}

In \cite{arora2008unique}, they first solve a semidefinite program and
then they use its solution to produce a new set of permutations
$\{\sigma_{uv}\}$ for each edge $uv \in E$.  Suppose that initial
instance (on which the SDP is solved) is on a complete graph and is
$(1-\eps)$-satisfiable.  If the new instance (using the
$\sigma$-permutations) is also $(1-\eps)$-satisfiable (e.g., if
$\sigma_{uv} = \pi_{uv}$ for every edge), then their algorithm
produces the same output as the Pivot Algorithm
and the loss function $f(\eps) = 3 \eps$.  
The analysis used in
\cite{arora2008unique} does not seem sufficient to show that the new
instance on the $\sigma$-permutations is actually a
$(1-\eps')$-satisfiable instance for some $\eps' < \eps$.  Thus, it
seems that new analysis or modifications of the algorithm is necessary
to obtain an improved loss function.

%% file: voting-algorithm.tex
\section{The Voting Algorithm}
\label{sec:voting}

In this section we present the voting algorithm for
\minLinEqQcomp, and show that this algorithm is a super robust approximation
algorithm for \minLinEqQcomp.  The idea is to begin with the pivot
algorithm from the previous section and use the resulting labels as
``temporary'' labels.  Then, we ``correct'' this labeling: each vertex
(except the pivot) casts a ``vote'' for the label of every other
vertex according to the relevant constraint.  The votes are tallied
for each vertex by a plurality rule: the final label of a vertex is
one that occurs most often in the list of its votes. The algorithm,
which we call the {\bf{Voting Algorithm}} is presented formally
in Section \ref{sec:n3-alg}.  The runtime of the Voting Algorithm
is $O(n^3)$.  In Section \ref{sec:n2-alg}, we present an algorithm
that is equivalent to the Voting Algorithm in that it produces the
same output assignment.  In Section \ref{sec:randomized}, we present a
randomized version of the Voting Algorithm with running time $O(n^2)$
and a slightly worse approximation guarantee.

\subsection{The Voting Algorithm for \minLinEqQcomp}\label{sec:n3-alg}

\vspace{5mm}
\noindent
\fbox{\parbox{15cm}{

{\bf Voting Algorithm} 

\vspace{1mm}

{\it Input:} An instance of \minLinEqQcomp.

\begin{itemize}
\item[1.] Pick a pivot $p \in V$. Label $p$ with $0$ and label each
  vertex $v \in V$ with temporary label $\TMP(v)$, which is chosen
  according to the constraint on edge $(p, v)$. (Specifically, $\TMP(v)
  = c_{vp}$.)
\item[2.] For each vertex $v$, each neighboring vertex $u \neq p$
  votes for a label for $v$, where $u$'s vote is based on its
  temporary label $\TMP(u)$.  (Specifically, the vote of $u$ for $v$
  is $(c_{vu} + \TMP(u)) \bmod q$.)
\item[3.] Then each $v$ is assigned a final label $\FIN(v)$ according
  to the outcome of its $n-2$ votes (with a plurality rule).  Ties are
  resolved arbitrarily.
\item[4.] Output the best $\FIN$ solution over all choices of $p$ in Step 1.
\end{itemize}}}

Notice that for technical reasons, we do not let the pivot
$p$ vote in Step 2.

\begin{theorem}\label{T:main}
On a $(1-\eps)$-satisfiable instance of \minLinEqQcomp, for $0 \leq
\varepsilon < \frac{1}{2}$, the Voting Algorithm returns a solution
with at most $(\eps+ c_{\varepsilon}\varepsilon^2) m$ unsatisfied
constraints where $\lim_{\varepsilon \to 0 } c_{\varepsilon}^{} = 16$.
\end{theorem}

An intuition of the proof is as follows. After Step 2., it is easy to see that an assignment is obtained where an $\varepsilon$ fraction of the vertices are incorrect (compared to the optimal solution). This does not ensure that a small enough fraction of the \emph{edges} are incorrect (i.e., unsatisfied), which is our goal. Therefore, we add the voting step in Step 3., which drastically reduces the number of unsatisfied edges towards our stated goal. This works because in order for a vote to produce a wrong assignment at a vertex, there needs to be a sizable number of either incorrect voters or incorrect adjacent edges, which we can control using a simple charging scheme.

We prove Theorem \ref{T:main} via the following lemma.
\begin{lemma}\label{L:mainlemma}
The Voting Algorithm returns a solution with at most $(\eps + 2
\eps^2 \nu (2+\nu) + o(1)) m$ unsatisfied constraints, where
$\nu=2/(1-2\varepsilon)$.
\end{lemma}

Fix an optimal solution $\opt$ and denote by $\OPT(v)$ the label it gives to a
vertex $v$.  In this fixed optimal solution, there are satisfied edges,
which we call \emph{green} edges and unsatisfied edges, which we call
\emph{red} edges.

Since $\varepsilon=\optval/m$, the number of red edges incident to $p$ is
at most $\varepsilon (n-1)$ for some choice of $p$.  We analyze the
voting algorithm for this choice of $p$.  Without loss of generality,
we assume that $\OPT(p)=0$.  This means that at least
$(1-\varepsilon)(n-1)$ vertices have $\TMP(u)=\OPT(u)$; we call these
\emph{good vertices} (i.e., incident to green edges), while the other
ones are \emph{rogue vertices} (i.e., incident to red edges).
 See Figure
\ref{fig:GoodRogue} for an illustration.

\begin{figure}[t]
\begin{center}
\includegraphics[width=0.5\textwidth]{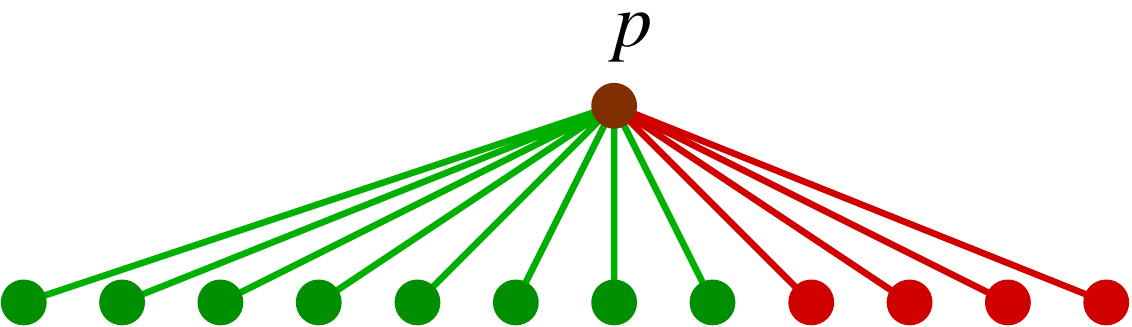}
  \caption{Green and red edges are those satisfied and unsatisfied, respectively, in $\OPT$.  Green and red vertices are good vertices and rogue vertices, respectively.}
  \label{fig:GoodRogue}
\end{center}
\end{figure}

The plan is to analyze how much the outcome of the voting algorithm
differs from \opt. A vertex is \emph{flipped} if $\FIN(v) \neq
\OPT(v)$. For a vertex to be flipped, it must be badly influenced by
its neighbors.  Let $\delta(v) \subset E$ denote the edges incident to
vertex $v$.  Observe that all good vertices adjacent to $v$ via a
green edge in $\delta(v)$ vote correctly with respect to vertex $v$
(i.e., they vote for label $\OPT(v)$).

The two types of vertices that can vote incorrectly for $v$'s label
(i.e., they might not vote for label $\OPT(v)$) are (i) rogue vertices
incident to green edges in $\delta(v)$, and (ii) vertices incident to
red edges in $\delta(v)$.  The number of vertices falling into the
first category is at most the number of rogue vertices (i.e., at most
$\varepsilon (n-1)$).  The number of vertices falling into the second
category is at most the number of red edges incident to $v$.  Hence we
say that a vertex $v$ is \emph{flippable} if the number of red edges
incident to $v$ is at least $(n-1)/2-\varepsilon (n-1)$.

\begin{lemma}\label{L:flips}
If a vertex $v$ is not flippable, it is not flipped (i.e., $\FIN(v) = \OPT(v)$).
\end{lemma}

\begin{proof}
A non-flippable vertex $v$ has at least $(n-1)/2+\varepsilon (n-1)+1$
incident green edges (since by definition the number of incident red
edges is at most $(n-1)/2 - \varepsilon (n-1) -1$).  At least
$(n-1)/2+1$ of these edges are incident to good vertices.  (Recall a
vertex $u$ is good if $\TMP(u)=\OPT(u)$.)  Thus all of these good
vertices vote for $v$ to be labeled $\OPT(v)$, and they will win the
vote since they form an absolute majority.
\end{proof}

\begin{lemma}\label{L:count}
There are at most $\varepsilon \nu n$ flippable vertices.
\end{lemma}

\begin{proof}
By definition, there are $\optval=\varepsilon m$ red edges. Denote by $f$
the number of flippable vertices.  Summing the red degree around each
flippable vertex gives $f \cdot ((n-1)/2-\varepsilon (n-1)) \leq
2\varepsilon m$ implying the lemma.
\end{proof}

At the end of the algorithm (i.e., according to the labels
$\{\FIN(v)\}$), if an edge is unsatisfied, then either it is red, or
it is green and at least one of its endpoints got flipped. In the
latter case, we charge that edge positively to (one of) the
endpoint(s) that got flipped. Similarly, if an edge is satisfied, then
either it is green, or it is red and at least one of its endpoints got
flipped. In the latter case, we charge that edge negatively to (one
of) the endpoint(s) that got flipped.

\begin{lemma}\label{L:charges}
The charges on a flipped vertex $v$ at the end of the algorithm are at
most $2\varepsilon(n-1) + \varepsilon\nu n$.
\end{lemma}

\begin{proof}
For a given vertex $v$, each neighbor $u$ votes for vertex $v$ to have the
label $vote(u \rightarrow v)$, where $vote(u \rightarrow v)$ is equal
to $\TMP(u)$ modified according to the constraint on the edge $uv$. 
A \emph{coalition} is a maximal set of neighboring vertices $C$
adjacent to $v$ that vote unanimously: for all $u \in C$, $vote(u
\rightarrow v)$ has the same value.
All the vertices adjacent to $v$ get partitioned into coalitions, and the
\emph{winning coalition} is one with the largest cardinality.

\begin{figure}[t]
\begin{center}
  \epsfig{file = 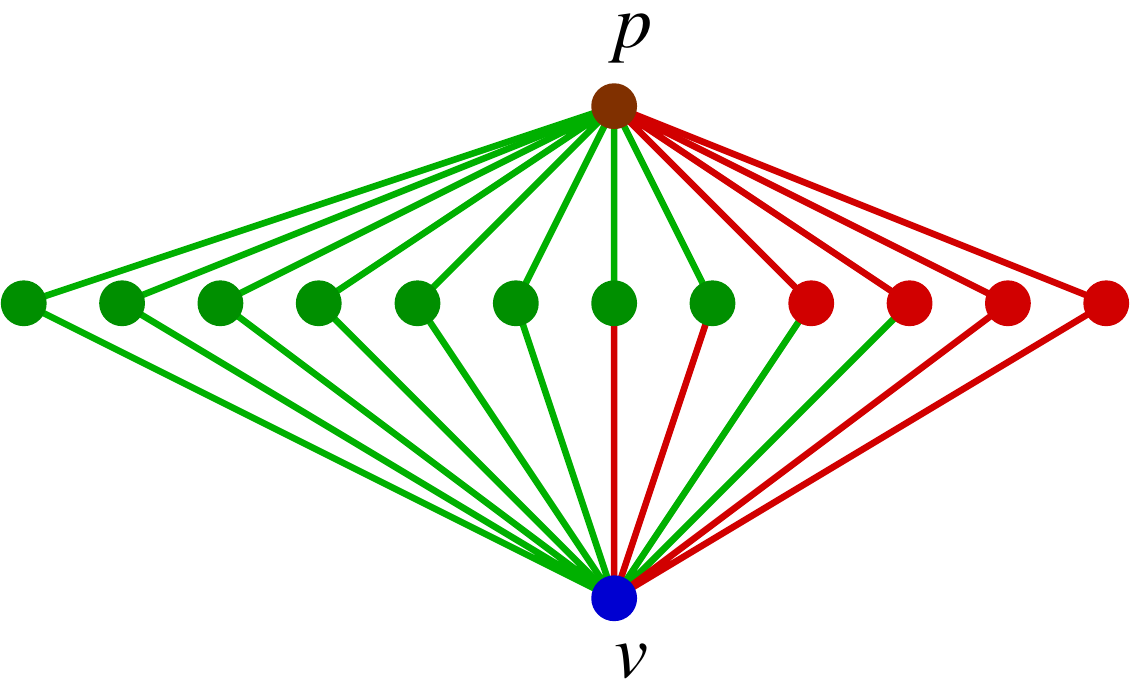, width=6cm} 
\epsfig{file = 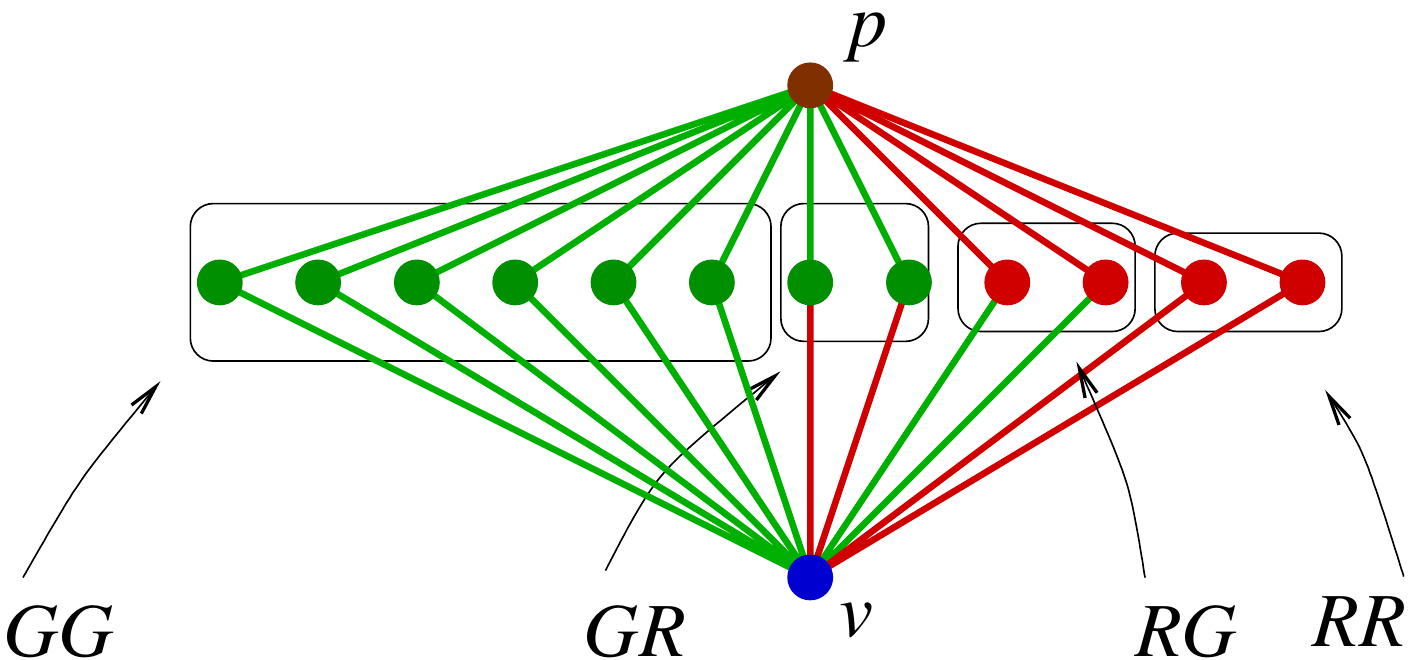, width=7.5cm} 
  \caption{There are $n-2$ vertices that vote for the label of $v$.
    They are partitioned into four sets: $GG$ are good vertices incident to
    green edges in $\delta(v)$; $GR$ are good vertices incident to red
    edges in $\delta(v)$; $RG$ are rogue vertices incident to green
    edges in $\delta(v)$; $RR$ are rogue vertices incident to red
    edges in $\delta(v)$. }
  \label{fig:GG}
\end{center}
\end{figure}
  
A flippable vertex $v$ gets flipped if the winning coalition $C_{WIN}$
is not the coalition $C_{OPT}$ (where $C_{OPT}$ is the coalition that
votes for $\OPT(v)$).  Observe that $C_{OPT}$ contains the subset of
good vertices that are incident to green edges in $\delta(v)$.  Call
this subset $GG$.  (See Figure \ref{fig:GG}.)  The winning coalition
$C_{WIN}$ is formed of good vertices incident to red edges in
$\delta(v)$ (call this subset $W_{GR}$), rogue vertices incident to
green edges in $\delta(v)$ (call this subset $W_{RG}$), and rogue
vertices incident to red edges in $\delta(v)$ (call this subset
$W_{RR}$).  (Observe that $W_{GR} \subseteq GR, W_{RG} \subseteq RG$
and $W_{RR} \subseteq RR$.  Moreover, note that there might be some
vertices in $V \setminus{\{p,v\}}$ that belong to neither $C_{OPT}$
nor to $C_{WIN}$.)

Since the winning coalition wins the vote, $|C_{OPT}| \leq |C_{WIN}|$. Thus,
$$|GG| \leq |W_{GR}| + |W_{RG}| + |W_{RR}| \leq |W_{GR}| +
\varepsilon(n-1).$$ The positive charges are upper bounded by
$|GG|+|RG|$.
(This is not an equality as these
edges might end up satisfied if their other endpoint is flipped as
well.)  
The negative charges are at least $|W_{GR}|$ minus
  those whose other endpoint has been flipped as well.  For the other
endpoint to be flipped, it needs to be flippable, so the total number
of negative charges is at least $|W_{GR}|- \varepsilon \nu n$.

So the total charge is at most:
\begin{eqnarray*}
|GG|+|RG|-(|W_{GR}|- \varepsilon \nu n) & \leq & |W_{GR}| + \varepsilon(n-1)
+ |RG| - |W_{GR}| + \varepsilon \nu n \\
& = &
\varepsilon(n-1) + |RG| + \varepsilon \nu n \\
& \leq &
2 \varepsilon (n-1) + \varepsilon \nu n,
\end{eqnarray*}
where we used the fact that $RG$ is a subset of the
rogue vertices and therefore has cardinality at most $\varepsilon(n-1)$.
\end{proof}

Denote by $\VAL$ the number of unsatisfied edges at the end of the algorithm.

\begin{lemma}\label{lem:last}
$\VAL-\optval \leq (1 + o(1))\cdot\optval \cdot 2\varepsilon \nu (2+\nu)$.
\end{lemma}

\begin{proof}
  This difference is exactly the number of green edges (i.e.,
  satisfied in $\OPT$) which become unsatisfied in $\VAL$ minus the number
  of red edges (i.e., unsatisfied in $\OPT$) which become satisfied in
  $\VAL$. This difference is exactly controlled by the charging
  scheme. Combining with Lemmas~\ref{L:count} and~\ref{L:charges},
  the sum of all charges is at most 
  \begin{align*}
	 \varepsilon\nu n(2\varepsilon(n-1) + \varepsilon\nu n) &=  2\varepsilon^2\nu n(n-1) + \varepsilon^2\nu^2 n^2 \\
	 &=  2\nu \varepsilon^2 \frac{n(n-1)}{2} (2+ \nu
         \frac{n}{n-1}) \\
	 &=  \eps \frac{n(n-1)}{2} (4\nu \eps + 
         2 \nu^2\eps \frac{n}{n-1}) \\
 	 & = \optval \cdot 2\eps \nu(2 + \nu (1 +
         \frac{1}{n-1}))\\
 	 & \leq \optval \cdot 2\eps \nu (2 + \nu) \cdot (1 +
         \frac{1}{n-1}).
  \end{align*}

\end{proof}

We note that Lemma \ref{L:mainlemma} is implied by Lemma \ref{lem:last}.

\subsection{Equivalent Implementation of Voting Algorithm}\label{sec:n2-alg}

We now give an equivalent interpretation of the Voting Algorithm.
Recall that $G=(V,E)$ is a simple, complete graph.  We define the
multigraph $\Gm$ to be a graph that contains $n-2$ edges connecting
$u$ and $v$, each edge corresponding to a path $uwv$ for $w \in
V\setminus{\{u,v\}}$.  Each new edge corresponding to a path $uwv$
inherits a $c_{uv}$ value from this path (i.e., $c_{uv} = (c_{uw} +
c_{wv}) \bmod q$).  Now we create a simple, complete graph $\Gs$ on
the vertex set $V$ in which the edge label $c_{uv}$ for edge $uv$ is
determined by taking the most popular $c_{uv}$ value from the $n-2$
values in $\Gm$ (ties broken arbitrarily).  Notice that it takes
$O(n^3)$ time to construct the instance $\Gs$ of $\minLinEqQcomp$,
since it takes $O(n)$ time to compute the constraint value on an edge.

Now we can run the Pivot Algorithm from Section \ref{sec:pivot} on
the input instance $\Gs$, which takes $O(n)$ time to output an
assignment and takes $O(n^2)$ time if try every vertex as a pivot.
Notice that the best output of the Pivot Algorithm on $\Gs$ (over all
pivots) is
the same as the output of the Voting Algorithm on $G$.    

\subsection{A Faster Randomized Voting Algorithm}\label{sec:randomized}

Instead of trying all vertices to be the pivot in Step 1. of the
Voting Algorithm, we simply choose a single pivot uniformly at
random.  We refer to this as the {\bf Randomized Voting Algorithm}.

If we choose a vertex $v$ at random, by Markov's Inequality, it
has probability at least $1/2$ of being incident to at most $2\eps n$ red
edges in a fixed optimal solution.  Thus, we execute the analysis used
in Section \ref{sec:voting} replacing $\eps$ with $2\eps$ which leads
to the following theorem.

\begin{theorem}\label{T:main-rand}
On a $(1-\eps)$-satisfiable instance of \minLinEqQcomp, for $0 \leq
\varepsilon < \frac{1}{2}$, with probability at least $1/2$, the Randomized Voting Algorithm returns a solution
with at most $(\eps+ c_{\varepsilon}\varepsilon^2) m$ unsatisfied
constraints where $\lim_{\varepsilon \to 0 } c_{\varepsilon}^{} = 32$.
\end{theorem}

\subsection{Extension to \UGC}\label{sec:ext-ugc}

In the more general setting of \UGC, we cannot assume that for any vertex $v$
there is an optimal solution that assigns the label $0$ to $v$. We
modify the Voting Algorithm from Section~\ref{sec:voting} slightly to
take this into account and obtain the following result, which differs
from the case of \minLinEqQcomp(i.e., Theorem \ref{T:main}) only in the runtime.

\begin{theorem}\label{T:main:ugc}
On a $(1-\eps)$-satisfiable instance of \UGC, for $0 \leq
\varepsilon < \frac{1}{2}$, the Voting Algorithm returns a solution
with at most $(\eps+ c_{\varepsilon}\varepsilon^2) m$ unsatisfied
constraints where $\lim_{\varepsilon \to 0 } c_{\varepsilon}^{} = 16$.
The runtime of the algorithm is $O(qn^3)$.
\end{theorem}

The only necessary modification of the Voting Algorithm is in Step 1.
For each label $\ell \in [q]$ and each pivot choice $p$, the algorithm
assigns label $\ell$ to $p$ and then computes the \TMP and \FIN labels
as before (see Steps 1.--3.~of the Voting Algorithm). The algorithm
returns the \FIN labels with the fewest violated constraints.  Thus,
the runtime is multiplied by a factor of $q$.  The analysis of the
modified voting algorithm is identical to the analysis presented in
Section~\ref{sec:voting}, once we fix a pivot $p$ with label $\ell$,
such that the number of red edges incident to $p$ is at most
$\varepsilon(n-1)$.

\subsection{Extension to Everywhere-Dense Case}\label{sec:dense}

For a graph $G=(V,E)$, let $d(v)$ denote the degree of a vertex $v$.
Following \cite{arora1999polynomial}, we define an {\em everywhere
$(1-\delta)$-dense} graph $G=(V,E)$ to be a graph in which 
$d(v) \geq (1-\delta)(n-1)$ for each vertex $v \in V$.  We can extend the Voting Algorithm to this case.  The algorithm is slightly modified.

\vspace{5mm}
\noindent
\fbox{\parbox{15cm}{

{\sc Voting Algorithm for Everywhere-Dense Graph} 

\vspace{1mm}

{\it Input:} An instance of \minLinEqQ on a $(1-\delta)$-everywhere
dense graph $G=(V,E)$.

\begin{itemize}

\item[1.] Pick a pivot $p \in V$. Label $p$ with $0$, and label each
  vertex $v \in V$ adjacent to $p$ with temporary label $\TMP(v)$, which is chosen  according to the constraint on edge $(p, v)$. (Specifically, $\TMP(v)
  = c_{vp}$.)

\item[2.] For each vertex $v$, each neighboring vertex $u$ with
  a $\TMP$ label votes for a label for $v$, where $u$'s vote
  is based on its temporary label $\TMP(u)$.  (Specifically, the vote
  of $u$ for $v$ is $(c_{vu} + \TMP(u)) \bmod q$.)

\item[3.] Then each $v$ is assigned a final label $\FIN(v)$ according
  to the outcome of the votes it received (with a plurality rule).  Ties are
  resolved arbitrarily.

\item[4.] Output the best solution over all choices of $p$ in Step 1.

\end{itemize}}}

\vspace{3mm}

Notice that in contrast to the Voting Algorithm on a complete graph,
$p$ also votes in Step 2.  The algorithm would also work if $p$ does
not vote, but the analysis turns out to be cleaner if $p$ votes.

Let $\optval$ denote the value of an optimal solution (i.e., the
minimum number of unsatisfied constraints) and let $\optval = \eps m$
(i.e., $\eps = \optval/m$).  The proof of Theorem
\ref{T:maindense} is very similar to that of Theorem \ref{T:main} and
the details can be found in Appendix \ref{app:B}.

\begin{restatable}{theorem}{mainDenseCase}\label{T:maindense}
On a $(1-\eps)$-satisfiable, $(1-\delta)$-everywhere-dense instance of \UG, for $0 \leq
\varepsilon < \frac{1}{2}$, the Voting Algorithm returns a solution
with at most $m(\eps+ c_{\varepsilon}\eps^2)/(1-\delta)$ unsatisfied
constraints where $\lim_{\varepsilon \to 0 } c_{\varepsilon}^{} = 16$.
The runtime of the algorithm is $O(qn^3)$.
\end{restatable}

%% file: PTAS.tex
\section{PTAS}\label{sec:ptas}

The Voting Algorithm from Section \ref{sec:voting} provides a good
approximation to \UGC when the value of an optimal solution is small.
In the
opposing regime, when the value of an optimal solution is large, we can obtain a
good approximation for this solution by solving approximately the complementary problem \UGCMax, which is the problem of maximizing the number of satisfied constraints. This complementary problem is the maximization version of a Constraint Satisfaction Problem, and, when the alphabet size is constant, those admit very efficient approximation algorithms on dense graphs using sampling techniques, and thus also on complete graphs. 

In order to obtain a (randomized) polynomial-time approximation scheme (PTAS) for \UGC we rely on the following theorem, where we emphasize that $q$ is considered a constant (i.e., the $O(\cdot)$ notation hides an unspecified dependency on $q$). Note that the algorithm underlying this theorem (e.g., in~\cite{MS:08}) is a very simple greedy algorithm (but the analysis is not that simple).

\begin{theorem}[{\cite[Theorem 7]{karpinski2009linear}}]\label{T:Maxversion}
	For any Max-$2$-CSP and any $\tau > 0$ there is a randomized
        algorithm which returns an assignment of cost at least $OPT -
        \tau n^2$ in runtime $O(n^2) + 2^{O(1/\tau^2)}$.
\end{theorem}

A MAX-$2$-CSP is a CSP where each constraint involves two variables.
When the alphabet size is not constant, 
a general purpose PTAS for Max-CSPs on complete graphs is
  ruled out under Gap-ETH,
see Romero, Wrochna and \v{Z}ivn\'y~\cite[Corollary~E.5]{rwz}. Whether a PTAS exists for \UGCMax when the alphabet size is not constant seems to be open.

Our PTAS is then as follows.

\begin{theorem}\label{T:ptas}
When the alphabet size $q$ is constant, for any $\tau>0$, we
can compute a $(1+\tau)$-approximation for the problem \UGC in
time $O(n^3)+ 2^{O(1/\tau^4)}$.
  \end{theorem}

Note that the runtime in Theorem \ref{T:ptas} is $O(n^2) +
2^{O(1/\tau^4)}$ if we use the Randomized Voting Algorithm.
This is similar to a result of Karpinsky and Schudy~\cite{karpinski2009linear}, with a simpler algorithm.

\begin{proof}[Proof of Theorem~\ref{T:ptas}]
  Let $OPT$ denote the optimal value of the problem. If $2\nu (2+\nu)(OPT/m)<\tau$, where $\nu=2/(1-2OPT/m)$, then by Lemma~\ref{L:mainlemma} we get the needed approximation. Otherwise, since $\nu\geq 2$, we have $OPT \geq \tau m/16$, and thus $m\leq 16OPT/\tau$.

  In this case, we compute a $\tau'$-approximation to the complementary problem using Theorem~\ref{T:Maxversion}, for $\tau'=\tau^2/32$. This provides us with a solution where the number of satisfied edges is at least $(m-OPT)-\tau'n^2$, and thus the number of unsatisfied edges is at most $OPT+\tau'n^2\leq OPT+32\tau'  OPT/\tau \leq OPT (1+\tau)$.\end{proof}

The argument in the proof of Theorem~\ref{T:ptas} can be
  generalized as follows.

\begin{observation}
	Let \textsc{Min-CSP} denote a constraint satisfaction problem
        where the objective is to minimize the number of violated
        constraints, while \textsc{Max-Comp-CSP} denotes the
        complementary problem of maximizing the number of satisfied
        constraints.  If there exists a PTAS for \textsc{Max-Comp-CSP} and a super robust algorithm for \textsc{Min-CSP}, then there exists a PTAS for \textsc{Min-CSP}. 
\end{observation}

\begin{proof}
  As in the previous proof, we use one algorithm or the other depending on $OPT$, the optimal value of \textsc{Min-CSP}. We denote by $C$ the number of constraints and write $\varepsilon=OPT/C$.

  We fix any $\tau>0$, and if $\varepsilon<\tau$, then a super robust algorithm computes a solution of value $(\varepsilon+O(\varepsilon^2))C \leq OPT(1+O(\tau))$, i.e., up to rescaling $\tau$ by a constant factor we get the required approximation guarantee. Otherwise, we have $\tau\leq \varepsilon$, and running the PTAS for \textsc{Max-Comp-CSP} with a target approximation factor of $\tau'=\tau^2$ yields a solution of value at least $(C-OPT)(1-\tau^2)$, and thus the number of unsatisfied constraints is at most $OPT+C \tau^2 \leq OPT+C \tau OPT/C \leq OPT(1+\tau)$. 
\end{proof}

As a corollary of this observation, since the special case of
\textsc{Correlation Clustering} known as {\sc MinDisagree} on complete graphs
 is APX-hard and its complementary max version admits a PTAS~\cite{bansal2004correlation}, it is very unlikely to admit a super robust algorithm.

%% file: NPhardness.tex
\section{\NP-Hardness}\label{sec:nphard}

In this section we prove the following hardness results. First, we
prove standard \NP-hardness for the more general problem of \UGC.  The
proof for this theorem is similar in spirit to the hardness reductions
of \textsc{MinDisAgree}[$k$] by Giotis and
Guruswami~\cite{giotis2006correlation}.

\begin{theorem}\label{thm:NP-hardness}
\minLinEqQcomp is \NP-hard for $q=2$, and
\UGC is \NP-hard for any value of $q$.
\end{theorem}

While we do expect \minLinEqQcomp to be \NP-hard for values of $q \geq
3$, this does not seem to follow from these proof techniques, which
leverage the use of non-cyclic permutations.  Theorem
\ref{thm:random-hardness} provides a hardness proof for \minLinEqQcomp
using randomized reductions.  Notice that Theorems
\ref{thm:NP-hardness} and \ref{thm:random-hardness} are incomparable.

\begin{theorem}\label{thm:random-hardness}
Unless $\NP \subseteq \BPP$, \minLinEqQcomp has no polynomial-time algorithm.
  \end{theorem}

To prove Theorem \ref{thm:random-hardness}, we follow the general
approach used for {\sc Feedback-Arc-Set-Tournament}
in \cite{ailon2008aggregating,alon2006ranking,charbit2007minimum} and
for fully-dense CSPs on a binary domain~\cite{ailon2007hardness}. The proofs of both Theorem~\ref{thm:NP-hardness} and Theorem~\ref{thm:random-hardness} are deferred to Appendix~\ref{app:hardness}.

%% file: appendixA.tex
\appendix

\section{Analysis of Pivot Algorithm}\label{app:first}

In a given instance of \minLinEqQcomp on a graph $G=(V,E)$, 
each cycle in $G$ is either {\em consistent} or {\em inconsistent}.
Let $\I$ denote the set of inconsistent cycles and let $\T \subseteq
\I$ denote the set of inconsistent triangles in $G$.  Observe that a
feasible solution to Problem \ref{linEqModq} is a hitting set for the
set of inconsistent cycles.  Consider the following linear programming
relaxation of Problem \ref{linEqModq} and its dual.
\begin{align*}
\min \sum_{e \in E} & x_e\\
\sum_{e \in C} x_e & \geq 1 \text{ for all cycles } C \in \I,\\
x_e & \geq 0. \tag{$P_{UG}$}\label{ug}
\end{align*}

\begin{align*}
\max \sum_{C \in \I} & y_C\\
\sum_{C \in \I: e \in C} y_C & \leq 1 \text{ for all } e \in E,\\
y_C & \geq 0. \tag{$D_{UG}$}\label{ug_dual}
\end{align*}

\begin{claim}
Any fractional packing of inconsistent triangles in $G$ is a lower
bound on the optimal value of \ref{ug}.
\end{claim}

\begin{cproof}
The optimal value of \ref{ug} is lower bounded by a fractional packing
of inconsistent cycles (i.e., a feasible solution for \ref{ug_dual}).
A fractional packing of inconsistent triangles is a lower bound on a
fractional packing of inconsistent cycles.
\end{cproof}

\SimpleDualProof*

\begin{proof}
The Pivot Algorithm assigns a label $\ell(v) \in [q]$ to each $v \in
V$.  Each edge $uv \in E$ whose constraint is unsatisfied by the
labels $\ell(u)$ and $\ell(v)$ is added to the ``deletion set'' $F
\subset E$.  Let $G'$ be the graph consisting of the remaining edges
(i.e., $G' = (V, E\setminus{F})$).  The following claim follows directly from the definition of $G'$.

\begin{claim}
$G'$ contains no inconsistent cycles.
\end{claim}

Let $t = \{i,j,k\}$ be an inconsistent triangle in $G$ and let $A_t$ denote
the event that $p \in \{i,j,k\}$.  Let $p_t$ be the probability of
event $A_t$.  Then,
\begin{eqnarray}
\Ex[\text{Number of deleted edges}] & = & \sum_{t \in \T} p_t.\label{output-alg}
\end{eqnarray}

\begin{claim}\label{dual-LB}
Setting $y'_C = y'_t = \frac{p_t}{3}$ if $C = t \in \T$ and $y'_C = 0$
otherwise is dual feasible.
\end{claim}

\begin{cproof}
Let $B_e$ be the event that edge $e$ was deleted by the algorithm.
Let $B_e \wedge A_t$ be the event that edge $e$ was deleted due to
$A_t$.  Given event $A_t$, each edge in $t$ is equally likely to be
deleted.  So we have
\begin{align*}
Pr(B_{e} \wedge A_{t}) &= Pr(B_{e} | A_{t})Pr(A_{t}) \\	
		       &=\frac{1}{3} \times p_{t} \\
		       &= \frac{p_{t}}{3}. 
\end{align*}
Note that for any $t \neq t' \in \T$ such that $e \in t$ and $e \in
t'$, $B_{e} \wedge A_{t}$ and $B_{e} \wedge A_{t'}$ are disjoint
events.  Hence, $\underset{t: e \in t }{\sum} \Pr(B_{e} \wedge A_{t})
\le 1$.  This implies that, for all $e \in E$:
\begin{eqnarray*}
\sum_{C:e \in C} y'_C = \sum_{t:e\in t} \frac{p_t}{3} \leq 1.
\end{eqnarray*}
We can therefore conclude that $\{y'_C\}$ is a dual-feasible solution.
\end{cproof}

From Claim \ref{dual-LB} and \eqref{output-alg}, we can conclude that the
pivot algorithm has an approximation ratio of $3$.
\end{proof}

To derandomize the pivot algorithm, observe that we can run the
algorithm $n$ times, each time choosing a different vertex as pivot.
Consider some fixed optimal solution $\OPT$ that violates exactly
$\eps {n \choose 2} = \eps m = \optval$ constraints.  For some choice
of pivot, the number of labels the algorithm incorrectly assigns is at
most $\eps(n-1)$.  Since each of these vertices is incident to at most
$(n-1)$ edges, the total number of incorrect edges is at most
\[\eps m + \eps(n-1)^2 \leq 3 \eps m = 3\cdot \optval.\]

\subsection{Tight example}

We can show that the analysis yielding a 3-approximation ratio is
tight.  Imagine that we have a complete graph such that all edges
except those in a Hamilton cycle are
associated with the constraint 
$x_u - x_v \equiv 0 \bmod q$.  The edges in the Hamilton cycle are
associated with the constraint $x_u - x_v \equiv 1 \bmod q$.  Notice
that all
pivots lead to the same number constraints being (un)satisfied.  An
optimal solution can satisfy ${n \choose 2}-n$ constraints and leaves
$n$ constraints unsatisfied.  Let $p$ be the pivot and let $p-1$ and
$p+1$ be its two neighbors on the Hamilton cycle.  Then the following
edges are unsatisfied:
\begin{enumerate}

\item The $n-4$ edges in the Hamilton cycle with neither endpoint in
  $\{p-1, p, p+1\}$.

\item The $n-4$ edges not in the Hamilton cycle with endpoint $p-1$.

\item The $n-4$ edges not in the Hamilton cycle with endpoint $p+1$.

\end{enumerate}

So, asymptotically, we have $3n$ unsatisfied edges, while an optimal
solution leaves only $n$ edges unsatisfied.

%% file: appendixB.tex
\section{Analysis of Voting Algorithm in Everywhere-Dense Case}\label{app:B}

In this section, we prove the following theorem.
\mainDenseCase*

For convenience, we restate the algorithm.  For simplicity, it is
stated for \minLinEqQ.  It can be extended to \UG on everywhere-dense
graphs by trying all labels in the first step (see Section
\ref{sec:ext-ugc}).

\vspace{5mm}
\noindent
\fbox{\parbox{15cm}{

{\sc Voting Algorithm for Dense Case} 

\vspace{1mm}

{\it Input:} An instance of \minLinEqQ on a $(1-\delta)$-everywhere
dense graph $G=(V,E)$.

\begin{itemize}

\item[1.] Pick a pivot $p \in V$. Label $p$ with $0$, and label each
  vertex $v \in V$ adjacent to $p$ with temporary label $\TMP(v)$, which is chosen  according to the constraint on edge $(p, v)$. (Specifically, $\TMP(v)
  = c_{vp}$.)

\item[2.] For each vertex $v$, each neighboring vertex $u$ with
  a $\TMP$ label votes for a label for $v$, where $u$'s vote
  is based on its temporary label $\TMP(u)$.  (Specifically, the vote
  of $u$ for $v$ is $(c_{vu} + \TMP(u)) \bmod q$.)

\item[3.] Then each $v$ is assigned a final label $\FIN(v)$ according
  to the outcome of the votes it received (with a plurality rule).  Ties are
  resolved arbitrarily.

\item[4.] Output the best \FIN solution over all choices of $p$ in Step 1.

\end{itemize}}}

\vspace{3mm}

Note that $p$ also votes in Step 2.  As mentioned earlier, the analysis turns out to be cleaner if $p$ votes (i.e., does not abstain).

Let $\optval$ denote the value of an optimal solution
(i.e., the minimum number of unsatisfied constraints) and let $\optval =
\eps m$ (i.e., $\eps = \optval/m$).

\begin{lemma}\label{L:mainlemmadense}
The Voting Algorithm on an everywhere $(1-\delta)$-dense graph gives a
$(1+ 2 \nu (2+\nu)\varepsilon/(1-\delta))$-approximation of the
optimal solution, where $\nu
= \frac{2}{1-2\varepsilon-2\delta}$.
\end{lemma}

Fix an optimal solution $\OPT$, and denote by $\OPT(v)$ the label it gives to
a vertex $v$.  In this optimal solution, there are satisfied edges,
which we call \emph{green} edges and unsatisfied edges, which we call
\emph{red} edges.

Since $\varepsilon=\optval/m$, the number of red edges incident to $p$ is
at most $\varepsilon \cdot d(p)$ for some choice of $p$.  We analyze
the voting algorithm for this choice of $p$.  Without loss of
generality, we assume that $\OPT(p)=0$.  This means that at least
$(1-\eps)d(p)$ vertices have $\TMP(u)=\OPT(u)$; we call
these \emph{nice vertices}.  The ones with $\TMP(u) \neq \OPT(u)$
are \emph{rogue vertices}.  The remaining vertices with no $\TMP$
label (because the edge is missing) are \emph{abstaining vertices}.
Observe that there are at most $\eps \cdot d(p)$ rogue vertices and $n
- d(p) - 1 \leq \delta(n-1)$ abstaining vertices.  By convention, we
say that $p$ itself is a nice vertex.

Let $r$ denote the number of rogue vertices (so $r \leq \eps \cdot
d(p)$).  Let $\Delta(v) \subset E$ denote the edges incident to vertex
$v$.  Let $\bd(v)$ denote the number of neighbors of vertex $v$ that
are non-abstaining.  Notice that $\bd(v)$ is the number of votes that
vertex $v \neq p$ receives.

The plan is to analyze how much the outcome of the voting algorithm
differs from $\OPT$. A vertex is \emph{flipped} if $\FIN(v) \neq
\OPT(v)$. For a vertex to be flipped, it must be badly influenced by
its neighbors.  Observe that all nice vertices adjacent to $v$ via a
green edge in $\Delta(v)$ vote correctly with respect to vertex $v$
(i.e., they vote for label $\OPT(v)$).

The two types of vertices that can vote incorrectly for $v$'s label
(i.e., they might not vote for label $\OPT(v)$) are (i) rogue vertices
incident to green edges in $\Delta(v)$, and (ii) vertices incident to
red edges in $\Delta(v)$.  The number of vertices falling into the
first category is at most the number of rogue vertices (i.e., at most
$r$).  The number of vertices falling into the second
category is at most the number of red edges incident to $v$.  Hence we
say that a vertex $v$ is \emph{flippable} if the number of red edges
incident to $v$ is at least $\bd(v)/2 - r$.

\begin{claim}\label{L:flipsdense}
If a vertex $v$ is not flippable, it is not flipped (i.e., $\FIN(v) = \OPT(v)$).
\end{claim}

\begin{cproof}
If $v$ is not flippable, it has at least $\bd(v)/2 + r + 1$ incident
green edges (since by definition the number of incident red edges is
at most $\bd(v)/2 - r -1$).  At least $\bd(v)/2 + 1$ of these green
edges are incident to nice vertices.  (Recall a vertex $u$ is nice if
$\TMP(u)=\OPT(u)$.)  Thus all of these nice vertices vote for $v$ to be
labeled $\OPT(v)$, and they will win the vote since they form an
absolute majority, since the maximum possible number of votes is
$\bd(v)$.
\end{cproof}

\begin{claim}\label{L:countdense}
There are $f \leq \varepsilon \nu n$ flippable vertices.
\end{claim}

\begin{cproof}
By definition, there are $\optval=\varepsilon m$ red edges. Denote by $f$
the number of flippable vertices.  For a flippable vertex $v$, we need
at least $\bd(v)/2 - r \geq (1-2\delta)(n-1)/2 - r$ red edges in
$\Delta(v)$.  Since $m \leq n(n-1)/2$, we have
$$  f \cdot ((1-2\delta) (n-1)/2- r) \leq
2\varepsilon m \leq \eps n (n-1).$$
Recall $r \leq \eps \cdot d(p) \leq \eps (n-1)$
which implies
\begin{eqnarray*}
f  
& \leq &
\frac{2\eps n (n-1)}{(1-2\delta) (n-1)- 2r}\\
 & \leq &
\frac{2\eps n (n-1)}{(1-2\delta) (n-1)- 2 \eps (n-1)}\\
 & \leq &
\frac{2\eps n}{1-2\delta - 2 \eps}.
\end{eqnarray*}
implying the lemma.
\end{cproof}

At the end of the algorithm (i.e., according to the labels
$\{\FIN(v)\}$), if an edge is unsatisfied, then either it is red, or
it is green and at least one of its endpoints got flipped. In the
latter case, we charge that edge positively to (one of) the
endpoint(s) that got flipped. Similarly, if an edge is satisfied, then
either it is green, or it is red and at least one of its endpoints got
flipped. In the latter case, we charge that edge negatively to (one
of) the endpoint(s) that got flipped.

\begin{claim}\label{L:chargesdense}
The charges on a flipped vertex $v$ at the end of the algorithm are at
most $2r + f \leq 2\varepsilon(n-1) + \varepsilon\nu n$.
\end{claim}

\begin{cproof}
For a given vertex $v$, each non-abstaining neighbor $u$ votes for
vertex $v$ to have the label $vote(u \rightarrow v)$, where
$vote(u \rightarrow v)$ is equal to $\TMP(u)$ modified according to
the constraint on the edge $uv$.
A \emph{coalition} is a maximal set of neighboring vertices $C$
adjacent to $v$ that vote unanimously: for all $u \in C$, $vote(u
\rightarrow v)$ has the same value.
All the non-abstaining vertices adjacent to $v$ get partitioned into
coalitions, and the
\emph{winning coalition} is one with the largest cardinality.

A flippable vertex $v$ gets flipped if the winning coalition $C_{WIN}$
is not the coalition $C_{OPT}$ (where $C_{OPT}$ is the coalition that
votes for $\OPT(v)$).  Observe that $C_{OPT}$ contains the subset of
nice vertices that are adjacent to $v$ via green edges.  Call this
subset $W_{GG}$.  The winning coalition $C_{WIN}$ is formed by nice
vertices adjacent to $v$ via red edges (call this subset $W_{GR}$),
rogue vertices adjacent to $v$ via green edges (call this subset
$W_{RG}$), and rogue vertices adjacent to $v$ via red edges (call this
subset $W_{RR}$).  (Note that there might be some vertices in $V
\setminus{v}$ that belong to neither $C_{OPT}$ nor to $C_{WIN}$, nor
to any coalition if they are abstaining vertices.)

Since the winning coalition wins the vote, $|C_{OPT}| \leq |C_{WIN}|$. Thus,
$$|W_{GG}| \leq |W_{GR}| + |W_{RG}| + |W_{RR}| \leq |W_{GR}| + r.$$
 The positive charges are upper bounded by $|W_{GG}|+|W_{RG}|$.  (This
 is not an equality as these edges might end up satisfied if their
 other endpoint is flipped as well.)  The negative charges are at
 least $|W_{GR}| + |W_{RR}|$ minus those whose other endpoint has been
 flipped as well and those incident to rogue neighbors (i.e.,
 $W_{RR}$).  For the other endpoint to be flipped, it needs to be
 flippable, so the total number of negative charges is at least
 $|W_{GR}|- f$.

So the total charge is at most:
\begin{eqnarray*}
|W_{GG}|+|W_{RG}|-(|W_{GR}|- \varepsilon \nu n) & \leq & |W_{GR}| + r
+ |W_{RG}| - |W_{GR}| + f \\
& = &
r + |W_{RG}| + f \\
& \leq &
2 r + f,
\end{eqnarray*}
where we used the fact that $W_{RG}$ is a subset of the
rogue vertices and therefore has cardinality at most $r$.
\end{cproof}

Denote by $\VAL$ the number of unsatisfied edges at the end of the algorithm.

\begin{claim}
$\VAL-\optval \leq \optval \cdot 2\varepsilon \nu (2+\nu)/(1-\delta)
+ \eps^2 \nu^2 n$.
\end{claim}

\begin{cproof}
  This difference is exactly the number of green edges (i.e.,
  satisfied in $\OPT$) which become unsatisfied in $\VAL$ minus the number
  of red edges (i.e., unsatisfied in $\OPT$) which become satisfied in
  $\VAL$. This difference is exactly controlled by the charging
  scheme. Combining with Claims~\ref{L:countdense} and~\ref{L:chargesdense},
  the sum of all charges is at most 
\begin{eqnarray}
(2r +f)f &\leq &
(2 \eps(n-1) + \eps \nu n) \eps \nu n\\ & = & (2 \eps(n-1) + \eps \nu (n
- 1)  + \eps \nu) \eps \nu n \\
& = & \eps^2 \nu n (n-1) (2 + \nu) + \eps^2 \nu^2 n \\
& \leq & \optval \cdot \frac{2\eps \nu (2 + \nu)}{1-\delta} + \eps^2 \nu^2 n.
\end{eqnarray}
Above we use the fact that
$$\OPT = \eps m \geq \eps n (1-\delta)(n-1)/2.$$
\end{cproof}

%% file: appendixhardness.tex
\section{Hardness proofs}\label{app:hardness}

In this section, we provide the proofs of Theorems~\ref{thm:NP-hardness} and~\ref{thm:random-hardness}.

\begin{proof}[Proof of Theorem \ref{thm:NP-hardness}]
We start with the \NP-hardness of \minLinEqQcomp for $q=2$. In that case,
we observe that the problem directly reduces from
\textsc{Correlation-Clustering} with a number of clusters fixed to be
$2$, which was studied by Giotis and
Guruswami~\cite{giotis2006correlation}. Precisely, Giotis and
Guruswami study the problem \textsc{MinDisAgree}[$k$], where one is
given a complete graph on $n$ nodes with each edge labelled by either
$+$ or $-$. The task is to partition the vertices into exactly $k$
clusters so as to minimize the number of $+$ edges between vertices in
different clusters, plus the number of $-$ edges between vertices in
the same cluster. For the special case $k=2$, this can be easily
encoded as a \minLinEqQcomp constraint in the following way. Following the
notation in the introduction, edges labelled $+$ get assigned an
integer $c_{uv}=0$, while edges labelled $-$ get assigned an integer
$c_{uv}=1$. Then, $+$ edges in different clusters and $-$ edges in the
same cluster directly translate into linear equations being violated,
which concludes the proof.

For the \UG problem on complete graphs, we start with the same
reduction, and pad it using additional quite trivial groups of
nodes. More precisely, let $H$ be an instance of
\textsc{MinDisAgree}[$2$] on $n$ vertices, to which we add $q-2$
collections $G_3, \ldots, G_{q}$ of $M$ vertices each, where $M$ is to
be determined later. We denote by $\tau_q^i$ the cyclic permutation of
order $q$ mapping $j$ to $j+i$ modulo $q$, and by $\sigma$ a fixed
permutation on $q-1$ letters without fixed points. The edges and their
constraints are as follows, where the vertices of $H$ and $G_i$ are
numbered arbitrarily:
\begin{itemize}
\item Between two vertices $u$ and $v$ of $H$, we choose $\pi_{u,v}$ to permute the first two coordinates if the edge is a $-$, or to be the identity on these two coordinates if the edge is a $+$. The rest of the permutation is the identity.
\item Between two vertices $u$ and $v$ of the same collection $G_i$, we choose $\pi_{u,v}$ to be so that $\pi_{u,v}(i)=i$, and $\pi_{u,v}=\sigma$ for the other $q-1$ values (with the $i$th value skipped).
\item Between two vertices $u$ and $v$ of different collections $G_i$ and $G_j$, we choose $\pi_{u,v}$ to be $\tau_q^{j-i}$.
\item Between two vertices $u$ and $v$, where $u$ is in $H$ and $v$ is in $G_i$, we choose $\pi_{u,v}$ to be $\tau_q^{i}$ for half of the $v$ in $G_i$, and $\tau_q^{i-1}$ for the other half.
\end{itemize}

We claim that the optimal solution\footnote{There are actually two
  different solutions here, depending on which cluster gets labelled
  $0$ and $1$. They have the same cost and by a slight abuse, we
  consider them to be the same.} to this \UG instance is assigning $x_u=i$ for each vertex in $G_i$, and assigning $x_u=0$ for the vertices in one cluster of the \textsc{MinDisAgree}[$2$] instance in $H$, and $x_u=1$ for the other cluster. Denoting by $c$ the cost of the \textsc{MinDisAgree}[$2$] instance, the cost of this solution is exactly $OPT:=c+nM(q-2)/2$, with $c$ bounded by $\binom{n}{2}$.

The proof that any minimal solution has this structure is as follows. Let $\ell$ be a labeling for a minimal solution. We first claim that for any collection $G_i$, all the vertices in $G_i$ have the same label. For each $j \in [3,q]$, let $S_j$ denote the biggest set of vertices of $G_j$ having the same label. Note that any $S_j$ has size at least $M/q$, and thus all the vertices in $S_j$  must be labeled by $j$ since otherwise the labels between them are violated (as $\sigma$ has no fixed points), yielding $M^2/q^2$ violated constraints, which is bigger than $OPT$ for $M=\Omega(q^2n^2)$. Similarly, the size of the second biggest label in a $G_i$ is at most $M/100q$. If $u$ is a vertex in $G_i$ that is not labelled $i$, all the constraints between $u$ and all the $S_j$ are violated, and changing the label of $u$ so that it matches that of $S_i$ fixes at least these $(q-3) M/q$ constraints, breaks at most $(q-3) M/100q$ constraints between the $G_i$, and breaks at most $n$ constraints with vertices in $H$. So the number of violated constraints is reduced if $99(q-3)M/100 >n$, contradicting the minimality of $\ell$.

We now claim that the vertices in $H$ are labeled $0$ or $1$. Let $u$ be a vertex in $H$ that is not labeled $0$ or $1$. Then all of its constraints with all the $G_i$ are violated. Replacing its label by a $0$ or $1$ label might break up to $n-1$ constraints (within $H$) but fixes exactly half of the constraints with all the $G_i$, which gives a better solution for $M(q-2)/2>n-1$.

Since all the vertices in $H$ are labelled $0$ or $1$, the optimal solution corresponds directly to the optimal \textsc{MinDisAgree}[$2$] instance on $H$, which concludes the proof. \end{proof}

The proof of Theorem~\ref{thm:random-hardness} proceeds by ``blowing up'' an instance by replacing each vertex with $k$ copies. It starts with the following lemma, describing a particular bipartite gadget for each non-edge.

\begin{lemma}\label{lem:double-cloud}
For any positive integers $q$ and $\ell$, where $\ell > q$ and $\ell$
is a multiple of $q$, there exists an instance of \minLinEqQ on the
complete bipartite graph $K_{\ell,\ell}$ such that for any vertex
labeling of $K_{\ell,\ell}$, the total number of satisfied equations
is at least $\ell^2/q - \Theta(\ell^{\frac{3}{2}})$ and at most
$\ell^2/q +
\Theta(\ell^{\frac{3}{2}})$.
\end{lemma}

\begin{proof}
We orient all edges from one side of $K_{\ell,\ell}$ to the other
side.  For each of the $\ell^2$ arcs, we choose a label from
$[q]$ uniformly at random.  Notice that there are
$q^{2\ell}$ possible vertex labelings.

For any fixed labeling, the expected number of satisfied constraints
is $\mu = \ell^2/q$.  For a fixed labeling, let $X_{uv}$ denote the
random variable which is 1 if arc $(u,v)$ is satisfied by the randomly
chosen arc label (w.r.t.~the fixed vertex labeling) and 0 otherwise
and let $X = \sum_{uv \in E(K_{\ell,\ell})} X_{uv}$.

Recall some standard Chernoff bounds:
\begin{eqnarray*}
\Pr[X \geq (1 + \delta)\mu] & \leq e^{-\frac{\delta^2 \mu}{3}} \text{ and }
\Pr[X \leq (1 - \delta)\mu] & \leq e^{-\frac{\delta^2 \mu}{2}}.
\end{eqnarray*}
Let $B_1$ 
be the (bad) event that there is some vertex labeling for which the
number of satisfied constraints exceeds $(1+ \delta)\mu$, and let
$B_2$
be the (bad) event that there is some vertex labeling for which the
number of satisfied constraints is less than $(1- \delta)\mu$.

We have $\mu = \frac{\ell^2}{q}$. Setting $\delta = \sqrt{\deltUp}$,
where $c =60$, we have:
\begin{eqnarray*}
\Pr\left[X \geq \left(1 + \sqrt{\deltUp}\right)\frac{\ell^2}{q}\right]
 \leq& e^{-\left(\frac{\deltUp \ell^2}{3q}\right)} \\
\Pr\left[X \geq \mu + \sqrt{c \cdot \log{q}}\frac{\ell^{3/2}}{\sqrt{q}}\right]
 \leq& {e^{-20 \ell \cdot \log{q} }}.
\end{eqnarray*}
\begin{eqnarray*}
\Pr\left[X \leq \left(1 - \sqrt{\deltUp}\right)\frac{\ell^2}{q}\right]
 \leq& {e^{-\left(\frac{\deltUp \ell^2}{2q}\right)}} \\
\Pr\left[X \leq \mu - \sqrt{c \cdot \log{q}}\frac{\ell^{3/2}}{\sqrt{q}}\right]
 \leq& {e^{-30 \ell \cdot \log{q}}}.
\end{eqnarray*}

Now we take a union bound over all $q^{2\ell}$ vertex labelings.
We have
\begin{eqnarray*}
\Pr[B_1] + \Pr[B_2] \leq
\frac{e^{2\ell \log q}}{e^{20 \cdot \ell \log{q}}} +
\frac{e^{2\ell \log q}}{e^{30 \cdot \ell \log{q}}} 
= \frac{1}{e^{18 \ell \log{q}}} +\frac{1}{e^{28 \ell \log{q}}} < 1
\end{eqnarray*}

Thus, we can conclude that there is a positive probability that the
number of satisfied constraints is within the desired range and
therefore the necessary gadget exists.
\end{proof}

\begin{proof}[Proof of Theorem \ref{thm:random-hardness}]
    We begin with an arbitrary instance of \minLinEqQ on the graph
    $G=(V,A)$. (We can think of $G$ as an oriented graph.) For each
    arc $(u,v) \in A$, we have a constraint $x_u - x_v \equiv c_{uv}
    \bmod q$.  We pick an integer $k =\poly(n)$ whose exact value is
    determined later and where $n = |V|$ and $k$ is a multiple of $q$.
    We construct a new ``blown-up'' graph $G^k = (V^k, A^k \cup B^k
    \cup C^k)$ as follows: \begin{eqnarray*} V^k & = & \{v_i ~|~ v \in
      V, ~i \in \{1, \ldots, k\}\},\\ A^k & = & \{(u_i,v_j) ~|~(u,v)
      \in A, ~i,j \in \{1, \ldots, k\}\},\\ B^k & = & \{(u_i,v_j)
      ~|~(u,v) \notin A, ~i,j \in \{1, \ldots, k\}\},\\ C^k & = &
      \{(u_i,u_j) ~|~ u \in V, ~i \neq j \in \{1, \ldots,
      k\}\}.  \end{eqnarray*}

For a vertex $u \in V$, we refer to the corresponding $k$ copies
$\{u_1, \ldots, u_k\}$ in $V^k$ as a ``cloud''.  For an arc
$(u_i,v_j) \in A^k$ we use same constraint as $(u,v)$.  For an arc
$(u_i,u_j) \in C^k$ (i.e., an arc in a cloud), we can use the
constraint $c_{u_i u_j} = 0$.  For an arc $(u_i,v_j) \in B^k$, we use
the bipartite gadget constructed in Lemma \ref{lem:double-cloud}.

Let $B$
denote the set of non-arcs in $G$ (i.e., $|B| = {n \choose 2} - |A|$).
Let $val(H)$ denote the minimum number of unsatisfied constraints in
$H$ over all assignments $V(H) \rightarrow \{1, \ldots, q\}$.
We now relate the values $val(G)$ and $val(G^k)$.  We set $k
= \Omega(n^6)$.  Notice that in this case, $k^{\frac{3}{2}} \cdot |B|
= o(k^2)$.

We define $\Gstar$ to be the ``blow-up'' of $G$, which is a subgraph
of $G^k$.  Specifically, $\Gstar = (V^k, A^k \cup C^k)$.  We can use
$val(G^k)$ to estimate $val(\Gstar)$ via the following claim, which
follows from Lemma
\ref{lem:double-cloud}.

\begin{claim}
$$\left|val(G^k) - val(\Gstar) -k^2 \cdot
|B| \cdot \frac{q-1}{q}\right| = O(k^{\frac{3}{2}} \cdot |B|).$$
\end{claim}

Now we need to use $val(\Gstar)$ to compute $val(G)$.

\begin{claim}
$$val(\Gstar) \leq k^2 \cdot val(G).$$ 
\end{claim}  

\begin{cproof}
Consider an optimal vertex labeling for $G$ that leaves $val(G)$
constraints unsatisfied.  We can construct a solution for $\Gstar$
with the claimed upper bound.  For each vertex in $V$, assign the same
label to each vertex in the corresponding cloud in $V^k$.  Each
satisfied constraint in $G$ corresponds to $k^2$ satisfied constraints
in $\Gstar$.  Each unsatisfied constraint in $G$ corresponds to $k^2$
unsatisfied constraints in $\Gstar$.  Moreover, each cloud in $G^k$
has only satisfied constraints and contributes zero to $val(G^k)$.
\end{cproof}

\begin{claim}
$$k^2 \cdot val(G) \leq val(\Gstar).$$
\end{claim}  

\begin{cproof}
Consider an optimal vertex labeling for $\Gstar$ that leaves
$val(\Gstar)$ constraints unsatisfied.  We can construct a vertex
labeling for $G$ with the claimed upper bound.  To do this, for each
vertex $v \in V$, we sample a label uniformly at random from the $k$
vertices in $v$'s cloud.  Call this labeling $r:V \rightarrow \{1,
\ldots, q\}$.  Then $\Ex[val_r(G)] \leq val(\Gstar)/k^2$.  (In fact,
$\Ex[val_r(G)] = val(A^k)/k^2$.)  We can conclude that $val(G) \leq
val(\Gstar)/k^2$.
\end{cproof}

In conclusion, we can use $val(G^k)$ to determine $val(G)$.
\end{proof}